\documentclass[a4paper,12pt]{article}
\pdfoutput=1
\usepackage{graphicx,float}
\usepackage[bf,small]{caption}
\usepackage{amsmath}
\usepackage{amsfonts}
\usepackage{amsthm}
\usepackage{bbm}
\usepackage{amssymb,mathtools}
\usepackage{comment}
\usepackage[inline,shortlabels]{enumitem}
\setlist[enumerate,1]{label={\upshape(\roman*)}}
\usepackage[margin=1.25in]{geometry}
\usepackage[authoryear,round]{natbib}

\usepackage{setspace}
\usepackage{tikz}
\usepackage{subcaption}

\mathtoolsset{mathic=true}

\DeclareMathOperator*{\esssup}{ess\,sup}
\DeclareMathOperator*{\essinf}{ess\,inf}

\usepackage[colorlinks,allcolors=blue]{hyperref}
\usepackage{authblk}

\theoremstyle{plain}
\newtheorem{theorem}{Proposition}
\newtheorem{lemma}{Lemma}

\theoremstyle{definition}
\newtheorem{example}{Example}

\begin{document}
	\title{Optimal measure preserving derivatives revisited}
	\author{Brendan K.\ Beare}
	\affil{School of Economics, University of Sydney}
	\maketitle

\begin{center}	
Accepted for publication in \emph{Mathematical Finance}.
\end{center}

\begin{abstract}
	This article clarifies the relationship between pricing kernel monotonicity and the existence of opportunities for stochastic arbitrage in a complete and frictionless market of derivative securities written on a market portfolio. The relationship depends on whether the payoff distribution of the market portfolio satisfies a technical condition called adequacy, meaning that it is atomless or is comprised of finitely many equally probable atoms. Under adequacy, pricing kernel nonmonotonicity is equivalent to the existence of a strong form of stochastic arbitrage involving distributional replication of the market portfolio at a lower price. If the adequacy condition is dropped then this equivalence no longer holds, but pricing kernel nonmonotonicity remains equivalent to the existence of a weaker form of stochastic arbitrage involving second-order stochastic dominance of the market portfolio at a lower price. A generalization of the optimal measure preserving derivative is obtained which achieves distributional replication at the minimum cost of all second-order stochastically dominant securities under adequacy.
\end{abstract}

\onehalfspacing

\section{Introduction}\label{sec:intro}

In \citet{Beare11} I considered the problem of selecting a derivative security written on a market portfolio that achieves the minimum price among all derivative securities delivering the same payoff distribution as the market portfolio. The setting is based on the payoff distribution pricing model of \citet{Dybvig88}, and may be briefly summarized as follows. The sole source of uncertainty is the value of the market portfolio at the end of one period (this is a single period model), which we take to be a nonnegative and finite random variable $X$ governed by a probability measure $\mu$. Let $\Theta$ denote the collection of all Borel measurable maps from $\mathbb R_+$ to $\bar{\mathbb R}_+$ that are finite $\mu$-almost everywhere ($\mu$-a.e.). We regard any member of $\Theta$ as a derivative security available for purchase; in this sense, the market of derivative securities is complete. Each $\theta\in\Theta$ may be purchased at a price $p(\theta)$ equal to its integral with respect to a finite ``risk neutral'' measure $\nu$, i.e.\ at price $p(\theta)=\int\theta\mathrm{d}\nu$, and delivers a payoff of $\theta(X)$ at the end of one period. There is no meaningful notion of market frictions such as bid-ask spreads or fixed transaction costs in our model. We assume mutual absolute continuity of $\mu$ and $\nu$ (while allowing both $\mu$ and $\nu$ to be potentially discrete) and refer to the Radon-Nikodym derivative $\pi=\mathrm{d}\nu/\mathrm{d}\mu$ as the \emph{pricing kernel}. Mutual absolutely continuity should be regarded as a no-arbitrage condition, as it rules out the possibility of purchasing a security at a zero price that pays off with positive probability, or purchasing a security at a positive price that pays off with zero probability. However, as we will see shortly, it leaves open the possibility of a form of stochastic arbitrage.

We are concerned with the problem of choosing $\theta\in\Theta$ to minimize the price $p(\theta)$ subject to a constraint requiring the distribution of $\theta(X)$ to weakly dominate the distribution of $X$ in some sense. Four such constraints will be relevant: distributional equality, first-order stochastic dominance, concave stochastic ordering, and second-order stochastic dominance. We denote these constraints respectively by $\theta(X)\sim X$, $\theta(X)\gtrsim_1 X$, $\theta(X)\gtrsim_\mathrm{cv}X$, and $\theta(X)\gtrsim_2 X$. Definitions of the latter three orderings are provided at the beginning of Section \ref{sec:results}. If there exists a derivative security $\theta\in\Theta$ satisfying one of these four constraints while having price $p(\theta)$ less than the price of the market portfolio, which is $\int x\mathrm{d}\nu(x)$, then we regard this as a form of stochastic arbitrage.

In \citet{Dybvig88} it is assumed that the market portfolio payoff distribution $\mu$ is either atomless, or is comprised of a finite number of equally probable atoms. This assumption turns out to be important for distinguishing between the existence of different forms of stochastic arbitrage. Following the mathematical literature on function rearrangement \citep{Luxemburg67,Day70,ChongRice71} we refer to this property of $\mu$ as \emph{adequacy}. Let $F$ and $Q$ denote the (cumulative) distribution function and quantile function of the random variable $X$, and for $\theta\in\Theta$ let $F_\theta$ and $Q_\theta$ denote the distribution function and quantile function of the random variable $\theta(X)$. Under adequacy of $\mu$, \citet{Dybvig88} showed that the minimum constrained price
\begin{align}\label{eq:opteq}
	&\inf\{p(\theta):\theta\in\Theta,\,\theta(X)\sim X\},
\end{align}
called the distributional price of $X$, is equal to
\begin{align}\label{eq:Dybvigprice}
	&\int_0^1Q(u)Q_\pi(1-u)\mathrm{d}u.
\end{align}
Under atomlessness of $\mu$, and assuming that $\pi$ is not constant on any set of positive $\mu$-measure, I showed in \citet{Beare11} that the minimum constrained price is achieved by the derivative security
\begin{align}\label{eq:ompdold}
	\vartheta(x)&=Q(1-F_\pi(\pi(x))),
\end{align}
which I termed the \emph{optimal measure preserving derivative}.

We will say that the pricing kernel $\pi$ is \emph{monotone} if it is equal $\mu$-a.e.\ to a nonincreasing function. \citet{Dybvig88} showed under adequacy of $\mu$ that the distributional price of the market portfolio, given by \eqref{eq:Dybvigprice}, is less than the actual price of the market portfolio, given by $\int x\mathrm{d}\nu(x)$, if and only if the pricing kernel $\pi$ is not monotone. This result neatly characterizes the existence of stochastic arbitrage opportunities defined in terms of distributional equality when $\mu$ is adequate. The primary goal of this article is to clarify the extent to which it carries over to forms of stochastic arbitrage defined in terms of the orderings $\gtrsim_1$, $\gtrsim_\mathrm{cv}$ and $\gtrsim_2$, and the extent to which it depends on the adequacy of $\mu$.

Two recent articles motivate my pursuit of this goal. The first is \citet{PostLongarela21}. Theorem 1 therein establishes, in a general setting permitting an incomplete options market and nonzero bid-ask spreads for option prices, that a necessary and sufficient condition for the existence of a stochastic arbitrage opportunity defined in the sense of second-order stochastic dominance is the nonexistence of a monotone pricing kernel implying option prices within the observed bid-ask ranges. Specialized to our simpler setting of a complete and frictionless options market, the result may be stated as follows: If $\mu$ assigns all mass to a finite collection of points, then
\begin{align}\label{eq:opt2}
	&\inf\{p(\theta):\theta\in\Theta,\,\theta(X)\gtrsim_2X\}
\end{align}
is less than the market portfolio price $\int x\mathrm{d}\nu(x)$ if and only if the pricing kernel $\pi$ is not monotone. This is puzzling because the constraint in \eqref{eq:opt2} is weaker than the constraint in \eqref{eq:opteq}, and we know from \citet{Dybvig88} that the minimum constrained price in \eqref{eq:opteq} is less than the market portfolio price if and only if the pricing kernel is not monotone. Two questions naturally arise. The first is the extent to which the distinction between the two characterizations of pricing kernel monotonicity hinges on the differing conditions -- adequacy in \citet{Dybvig88} and finite support in \citet{PostLongarela21} -- placed on $\mu$. Neither of these conditions on $\mu$ implies the other. The second question is whether the constrained minimum in \eqref{eq:opt2} is attained by the optimal measure preserving derivative in \eqref{eq:ompdold}.

Propositions \ref{thm:pkSSD} and \ref{thm:adequate} in this article provide answers to these two questions. Proposition \ref{thm:pkSSD} shows, without imposing any condition on $\mu$ (beyond the no-arbitrage condition of mutual absolute continuity of $\mu$ and $\nu$), that the minimum constrained price in \eqref{eq:opt2} is less than the market portfolio price if and only if the pricing kernel is not monotone, and also that the minimum constrained price
\begin{align}\label{eq:optcv}
	&\inf\{p(\theta):\theta\in\Theta,\,\theta(X)\gtrsim_\mathrm{cv}X\}
\end{align}
is less than the market portfolio price if and only if the pricing kernel is not monotone. On the other hand, it is not true in general that the minimum constrained price in \eqref{eq:opteq} is less than the market portfolio price if and only if the pricing kernel is not monotone, and neither is it true in general that
\begin{align}\label{eq:opt1}
	&\inf\{p(\theta):\theta\in\Theta,\,\theta(X)\gtrsim_1X\}
\end{align}
is less than the market portfolio price if and only if the pricing kernel is not monotone. Under the assumption that $\mu$ is adequate, Proposition \ref{thm:adequate} shows that the minimum constrained prices in \eqref{eq:opteq}, \eqref{eq:opt2}, \eqref{eq:optcv} and \eqref{eq:opt1} are all equal to one another, are all less than the market portfolio price if and only if the pricing kernel is not monotone, and are all attained by a generalization of the optimal measure preserving derivative in \eqref{eq:ompdold}. That generalization takes the special form in \eqref{eq:ompdold} under the additional requirements that $\mu$ is atomless and that the pricing kernel $\pi$ is not constant on any set of positive $\mu$-measure, as imposed in \citet{Beare11}.

The second recent article motivating this one is \citet{KleinerMoldovanuStrack21}. The most relevant part is Section S.2.2 in the supplementary material. It is shown there that if $\mu$ is atomless, if $\pi$ is $\mu$-essentially bounded, and if $X$ is integrable, then the minimum constrained prices in \eqref{eq:opteq} and in \eqref{eq:optcv} are equal to one another, and are less than the market portfolio price if and only if the pricing kernel is not monotone. This result is contained in Proposition \ref{thm:adequate} given here, which relaxes atomlessness to adequacy and drops the essential boundedness and integrability conditions. The proof given by \citet{KleinerMoldovanuStrack21} relies upon the Hardy-Littlewood rearrangement inequality as in \citet{CarlierDana05} and in \citet{Beare11}, and also upon an inequality of \citet{FanLorentz54}, which they apply fruitfully to a range of problems beyond the one discussed here. The role played by the Fan-Lorentz inequality is replaced in the proofs of Propositions \ref{thm:pkSSD} and \ref{thm:adequate} given here by an older inequality due to Hardy.

The property of adequacy of a finite measure $\mu$ was originally defined in \citet{Luxemburg67} to mean that, given any nonnegative measurable functions $f_1$ and $f_2$, the set $\{\int f_1f_2'\mathrm{d}\mu:f_2'\sim f_2\}$ has a maximal element equal to the upper bound for $\int f_1f_2\mathrm{d}\mu$ given by the Hardy-Littlewood rearrangement inequality. It was observed that if a finite measure $\mu$ is atomless or is comprised of finitely many atoms of equal measure then it is adequate. It was subsequently shown in \citet{Day70} that a finite measure $\mu$ is adequate if and only if it is atomless or is comprised of finitely many atoms of equal measure. I take this simple characterization of adequacy to be its definition.

Specifications of $\mu$ in empirical applications which do not satisfy adequacy are common. For instance, the binomial option pricing model in \citet{CoxRossRubinstein79} specifies $\mu$ to be the distribution of an exponentially transformed binomial random variable, violating adequacy. Estimates of $\mu$ obtained by empirical likelihood methods, as in e.g.\ \citet{PostKarabatiArvanitis18}, will also not generally satisfy adequacy. On the other hand, specifications of $\mu$ not satisfying adequacy may nevertheless closely approximate smooth distributions, as in the case of the binomial approximation to the normal distribution. This raises the question of whether the claims established under adequacy in Proposition \ref{thm:adequate} are informative when $\mu$ is inadequate but closely approximates a suitably smooth adequate measure. We present some limited evidence in support of an affirmative answer to this question in a numerical illustration reported in Section \ref{sec:illustration}. The illustration involves a pair of sequences of inadequate measures $\mu_n$ and $\nu_n$ concentrated on $n$ atoms. The sequences are constructed in such a way that, as $n$ increases, the measures $\mu_n$ and $\nu_n$ more closely approximate atomless measures $\mu$ and $\nu$ corresponding to the estimated objective and risk-neutral distributions of monthly S\&P500 returns reported in \citet{Jackwerth00}. For each pair of measures $\mu_n$ and $\nu_n$ we use linear and mixed-integer linear programs to compute the price minimizing derivatives constrained to satisfy first- or second-order stochastic dominance over the market portfolio. We find that, when $n$ is large, both price minimizing derivatives are closely approximated by the optimal measure preserving derivative corresponding to the atomless measures $\mu$ and $\nu$. When $n$ is small, we observe substantial differences between the two price minimizing derivatives, with only the derivative constrained to satisfy second-order stochastic dominance closely approximating the optimal measure preserving derivative for the atomless case.

A referee has drawn my attention to a recent article by \citet{MagnaniRabanalRudWang22} reporting the outcome of laboratory experiments in which participants make portfolio choices in a simple binomial tree model with monotone pricing kernel. In three rows of Table 2 therein, the percentage of experimental participants selecting portfolios which satisfy one of three notions of efficiency is reported. The three notions of efficiency are the nonexistence of a lower cost portfolio which stochastically dominates the selected portfolio, with stochastic dominance defined in the first-, second- or third-order sense. Even though the binomial tree model does not satisfy our adequacy condition, we see that in each experimental design the percentages of efficient portfolio choices are identical under either the first- or second-order sense of efficiency. (There is one experimental design in which the reported percentage of efficient portfolio choices is 47\% using the first-order sense of efficiency and 48\% using the second-order sense, but this must be a rounding error as the former percentage can be no less than the latter by construction.) On the other hand, the percentages of efficient portfolio choices using the third-order sense of efficiency are substantially lower than those for the first- and second-order senses. It may be useful to pursue an extension of the results obtained in this article to the case of third-order stochastic dominance. I comment further on this possibility in the final paragraph of Section \ref{sec:discussion}.

The remainder of this article is structured as follows. The new results, Propositions \ref{thm:pkSSD} and \ref{thm:adequate}, are presented in Section \ref{sec:results}. The numerical illustration involving the approximation of an adequate measure by a sequence of inadequate measures is presented in Section \ref{sec:illustration}. In Section \ref{sec:discussion} I discuss the implications of my results for recent research on stochastic arbitrage under pricing kernel nonmonotonicity. In Appendix \ref{sec:mathprelim} I review for the reader's convenience various results on distributions, quantiles, rearrangements, and disintegration of measures that are used in the proofs of Propositions \ref{thm:pkSSD} and \ref{thm:adequate}, given in Appendix \ref{sec:proofs}.

\section{Results}\label{sec:results}
We begin by defining the relations $\gtrsim_1$, $\gtrsim_\mathrm{cv}$ and $\gtrsim_2$. Let $X$ and $X'$ be nonnegative and finite random variables. We say that $X\gtrsim_1X'$ if $\mathrm{E}u(X)\geq\mathrm{E}u(X')$ for all nondecreasing maps $u:\mathbb{R}_+\to\mathbb{R}$ such that the expectations exist. We say that $X\gtrsim_\mathrm{cv}X'$ if $\mathrm{E}u(X)\geq\mathrm{E}u(X')$ for all concave maps $u:\mathbb{R}_+\to\mathbb{R}$ such that the expectations exist. We say that $X\gtrsim_2X'$ if $\mathrm{E}u(X)\geq\mathrm{E}u(X')$ for all nondecreasing and concave maps $u:\mathbb{R}_+\to\mathbb{R}$ such that the expectations exist. Table B.1 in \citet{Dybvig88} provides equivalent reformulations of these definitions in terms of distribution functions, quantile functions, and additive decompositions. Note that if $X$ and $X'$ have finite expectations then $X\gtrsim_\mathrm{cv}X'$ if and only if $\mathrm{E}X=\mathrm{E}X'$ and $X\gtrsim_2X'$; see e.g.\ Section 3.A in \citet{ShakedShanthikumar07}. The concave stochastic order thus carries the interpretation of mean-preserving spread.

Our first result is as follows.
\begin{theorem}\label{thm:pkSSD}
	If $\mu$ and $\nu$ are mutually absolutely continuous with Radon-Nikodym derivative $\pi=\mathrm{d}\nu/\mathrm{d}\mathrm{\mu}$, then the following three conditions are equivalent.
    \begin{enumerate}
    	\item $\inf\{p(\theta):\theta\in\Theta,\,\theta(X)\gtrsim_\mathrm{cv}X\}<\int x\mathrm{d}\nu(x)$.\label{pkSSD1}
    	\item $\inf\{p(\theta):\theta\in\Theta,\,\theta(X)\gtrsim_2X\}<\int x\mathrm{d}\nu(x)$.\label{pkSSD2}
    	\item $\pi$ is not monotone.\label{pkSSD3}
    \end{enumerate}
\end{theorem}
The equivalence of conditions \ref{pkSSD2} and \ref{pkSSD3} in Proposition \ref{thm:pkSSD} is implied by Theorem 1 in \citet{PostLongarela21} under the additional requirement that $\mu$ has finite support. The equivalence of conditions \ref{pkSSD1} and \ref{pkSSD3} in Proposition \ref{thm:pkSSD} follows from the discussion in Section S.2.2 of the supplementary material to \citet{KleinerMoldovanuStrack21} under the additional requirements that $\mu$ is atomless, that $\pi$ is $\mu$-essentially bounded, and that $X$ is integrable.

Our second result is as follows.
\begin{theorem}\label{thm:adequate}
	If $\mu$ and $\nu$ are mutually absolutely continuous with Radon-Nikodym derivative $\pi=\mathrm{d}\nu/\mathrm{d}\mathrm{\mu}$, and if $\mu$ is adequate, then the following five conditions are equivalent.
	\begin{enumerate}
		\item $\inf\{p(\theta):\theta\in\Theta,\,\theta(X)\sim X\}<\int x\mathrm{d}\nu(x)$.\label{adequate1}
		\item $\inf\{p(\theta):\theta\in\Theta,\,\theta(X)\gtrsim_1X\}<\int x\mathrm{d}\nu(x)$.\label{adequate2}
		\item $\inf\{p(\theta):\theta\in\Theta,\,\theta(X)\gtrsim_\mathrm{cv}X\}<\int x\mathrm{d}\nu(x)$.\label{adequate3}
		\item $\inf\{p(\theta):\theta\in\Theta,\,\theta(X)\gtrsim_2X\}<\int x\mathrm{d}\nu(x)$.\label{adequate4}
		\item $\pi$ is not monotone.\label{adequate5}
	\end{enumerate}
    Moreover, the four infima in \ref{adequate1}, \ref{adequate2}, \ref{adequate3} and \ref{adequate4} are all equal to one another, and are all attained by the function
    \begin{align}\label{eq:ompd}
    	\vartheta(x)&=Q\Big(1-F_{\pi}(\pi(x))+\mu\{w:\pi(w)=\pi(x),w\leq x\}\Big).
    \end{align}
\end{theorem}

The equivalence of conditions \ref{adequate1} and \ref{adequate5} in Proposition \ref{thm:adequate} is shown in \citet{Dybvig88}. The equivalence of conditions \ref{adequate1}, \ref{adequate3}, \ref{adequate4} and \ref{adequate5} in Proposition \ref{thm:adequate} therefore follows from Proposition \ref{thm:pkSSD}. Condition \ref{adequate2} joins this group of equivalent conditions due to the fact that the relation $\gtrsim_1$ is implied by the relation $\sim$ and implies the relation $\gtrsim_2$.

The function $\vartheta$ in Proposition \ref{thm:adequate} is a more general version of the optimal measure preserving derivative introduced in \citet{Beare11} and stated above in \eqref{eq:ompdold}. The source of additional generality is the term $\mu\{w:\pi(w)=\pi(x),w\leq x\}$, which is equal to zero under the specificity that $\mu$ is atomless and $\pi$ is not constant on any set of positive $\mu$-measure, as imposed in \citet{Beare11}. Since $\gtrsim_2$ is the weakest and $\sim$ the strongest of the four relations $\sim$, $\gtrsim_1$, $\gtrsim_\mathrm{cv}$ and $\gtrsim_2$, to show that $\vartheta$ attains the infima in \ref{adequate1}, \ref{adequate2}, \ref{adequate3} and \ref{adequate4} it suffices to show that $\vartheta$ attains the infimum in \ref{adequate4} while also satisfying $\vartheta(X)\sim X$. Our proof that $\vartheta(X)\sim X$ extends the arguments in \citet{Beare11} by relaxing atomlessness to adequacy and dropping the nonflatness condition on $\pi$. The latter is achieved by applying the disintegration theorem discussed in \citet{ChangPollard97}, which provides a convenient way to handle the conditioning of probability measures on level sets of random variables.

I now present two examples that demonstrate the pivotal role played by the adequacy condition on $\mu$ in Proposition \ref{thm:adequate}. Both examples are very simple and involve a measure $\mu$ that assigns all mass to two atoms. The first example shows that \ref{adequate1} and \ref{adequate2} in Proposition \ref{thm:adequate} need not be equivalent if $\mu$ is not adequate, and that $\vartheta$ need not satisfy $\vartheta(X)\sim X$ if $\mu$ is not adequate.
\begin{example}\label{ex1}
	Suppose that $\{1\}$ and $\{2\}$ are atoms of $\mu$, with $\mu\{1\}=2/3$ and $\mu\{2\}=1/3$. Suppose further that $\nu$ is a probability measure with $\nu\{1\}=1/3$ and $\nu\{2\}=2/3$, so that $\int x\mathrm{d}\nu(x)=5/3$. For any $\theta\in\Theta$, the distribution of $\theta(X)$ assigns mass $2/3$ to the atom $\{\theta(1)\}$ and mass $1/3$ to the atom $\{\theta(2)\}$. We therefore have $\theta(X)\sim X$ if and only if $\theta(1)=1$ and $\theta(2)=2$, implying that condition \ref{adequate1} in Proposition \ref{thm:adequate} is not satisfied. On the other hand, any function $\theta\in\Theta$ with $\theta(1)=2$ and $\theta(2)=1$ satisfies $\theta(X)\gtrsim_1X$ and $p(\theta)=4/3<5/3$, so condition \ref{adequate2} in Proposition \ref{thm:adequate} is satisfied. Elementary calculations show that the function $\vartheta$ in Proposition \ref{thm:adequate} takes the values $\vartheta(1)=2$ and $\vartheta(2)=1$, so $\vartheta$ satisfies $\vartheta(X)\gtrsim_1X$ but not $\vartheta(X)\sim X$, and has price $p(\vartheta)=4/3$.
\end{example}

The second example shows that neither \ref{adequate1} nor \ref{adequate2} in Proposition \ref{thm:adequate} need be equivalent to \ref{adequate3}, \ref{adequate4} or \ref{adequate5} if $\mu$ is not adequate, and that $\vartheta$ may have price exceeding that of the market portfolio if $\mu$ is not adequate.
\begin{example}\label{ex2}
	Suppose that $\{1\}$ and $\{2\}$ are atoms of $\mu$, with $\mu\{1\}=1/3$ and $\mu\{2\}=2/3$. Suppose further that $\nu$ is a probability measure with $\nu\{1\}=1/5$ and $\nu\{2\}=4/5$, so that $\int x\mathrm{d}\nu(x)=9/5$. For any $\theta\in\Theta$, the distribution of $\theta(X)$ assigns mass $1/3$ to the atom $\{\theta(1)\}$ and mass $2/3$ to the atom $\{\theta(2)\}$. We therefore have $\theta(X)\sim X$ if and only if $\theta(1)=1$ and $\theta(2)=2$, implying that condition \ref{adequate1} in Proposition \ref{thm:adequate} is not satisfied. Furthermore, we have $\theta(X)\gtrsim_1X$ if and only if $\theta(1)\geq1$ and $\theta(2)\geq2$, implying that condition \ref{adequate2} in Proposition \ref{thm:adequate} is not satisfied. The pricing kernel $\pi$ satisfies $\pi(1)=3/5$ and $\pi(2)=6/5$, so $\pi$ is not monotone (i.e., not nonincreasing on $\{1,2\}$), and condition \ref{adequate5} in Proposition \ref{thm:adequate} is satisfied. Conditions \ref{adequate3} and \ref{adequate4} in Proposition \ref{thm:adequate} must therefore also be satisfied, by Proposition \ref{thm:pkSSD}. Elementary calculations show that the function $\vartheta$ in Proposition \ref{thm:adequate} takes the values $\vartheta(1)=\vartheta(2)=2$ and therefore satisfies $\vartheta(X)\gtrsim_1X$ but not $\vartheta(X)\sim X$, and has price $p(\vartheta)=2>9/5$.
\end{example}

\section{Numerical illustration}\label{sec:illustration}

To further illustrate our results we revisit the influential article by \citet{Jackwerth00} reporting an empirical violation of pricing kernel monotonicity with 31-day S\&P500 index options on April 15, 1992. The left panel in Figure \ref{fig:jackwerth} reproduces the estimated probability density functions for $\mu$ (solid line) and $\nu$ (dashed line) reported in Figure 2 in \citet{Jackwerth00}. The former density function was computed by applying a kernel density estimator to four years of past monthly S\&P500 returns, while the latter density function was estimated from prevailing option prices by applying a variation of the method of \citet{JackwerthRubinstein96}. The ratio of these density functions is the pricing kernel $\pi$, plotted in the center panel in Figure \ref{fig:jackwerth}. We see that $\pi$ is nonmonotone, exhibiting a pronounced increasing region around the center of the return distribution, and a smaller increasing region toward the left tail. The right panel in Figure \ref{fig:jackwerth} displays the optimal measure preserving derivative $\vartheta$ computed using the formula in \eqref{eq:ompdold} or in \eqref{eq:ompd}; the two formulae produce the same $\vartheta$ since $\pi$ does not have any flat regions. The plot of $\vartheta$ is identical to the one in Figure 4.3(b) in \citet{Beare11}. It deviates from the market payoff (the 45-degree line, dashed) due to the nonmonotone shape of $\pi$.

\begin{figure}
	\centering
		\begin{tikzpicture}
			\draw[thick, ->] (0,0) -- (3.5,0);
			\draw[thick, ->] (0,0) -- (0,3.3);
			\draw (0,0) -- ({-3pt},0) node[anchor=east,font=\scriptsize] {$0$};
			\draw (0,{3*(1/4)}) -- ({-3pt},{3*(1/4)}) node[anchor=east,font=\scriptsize] {$3$};
			\draw (0,{3*(2/4)}) -- ({-3pt},{3*(2/4)}) node[anchor=east,font=\scriptsize] {$6$};
			\draw (0,{3*(3/4)}) -- ({-3pt},{3*(3/4)}) node[anchor=east,font=\scriptsize] {$9$};
			\draw (0,{3*(4/4)}) -- ({-3pt},{3*(4/4)}) node[anchor=east,font=\scriptsize] {$12$};
			\draw (0,0) -- (0,{-3pt}) node[anchor=north,font=\scriptsize] {$.85$};
			\draw ({3.2*(1/4)},0) -- ({3.2*(1/4)},{-3pt}) node[anchor=north,font=\scriptsize] {$.925$};
			\draw ({3.2*(2/4)},0) -- ({3.2*(2/4)},{-3pt}) node[anchor=north,font=\scriptsize] {$1$};
			\draw ({3.2*(3/4)},0) -- ({3.2*(3/4)},{-3pt}) node[anchor=north,font=\scriptsize] {$1.075$};
			\draw ({3.2*(4/4)},0) -- ({3.2*(4/4)},{-3pt}) node[anchor=north,font=\scriptsize] {$1.15$};
			\node[rotate=90,font=\scriptsize] at (-.7,1.5) {Probability density};
			\node[font=\scriptsize] at (1.6,-.8) {Market payoff};
			\draw[blue,thick] plot file {physical.txt};
			\draw[red,thick,dashed] plot file {riskneutral.txt};
		\end{tikzpicture}
		\begin{tikzpicture}
			\draw[thick, ->] (0,0) -- (3.5,0);
			\draw[thick, ->] (0,0) -- (0,3.3);
			\draw (0,0) -- ({-3pt},0) node[anchor=east,font=\scriptsize] {$0$};
			\draw (0,{3*(1/4)}) -- ({-3pt},{3*(1/4)}) node[anchor=east,font=\scriptsize] {$.5$};
			\draw (0,{3*(2/4)}) -- ({-3pt},{3*(2/4)}) node[anchor=east,font=\scriptsize] {$1$};
			\draw (0,{3*(3/4)}) -- ({-3pt},{3*(3/4)}) node[anchor=east,font=\scriptsize] {$1.5$};
			\draw (0,{3*(4/4)}) -- ({-3pt},{3*(4/4)}) node[anchor=east,font=\scriptsize] {$2$};
			\draw (0,0) -- (0,{-3pt}) node[anchor=north,font=\scriptsize] {$.85$};
			\draw ({3.2*(1/4)},0) -- ({3.2*(1/4)},{-3pt}) node[anchor=north,font=\scriptsize] {$.925$};
			\draw ({3.2*(2/4)},0) -- ({3.2*(2/4)},{-3pt}) node[anchor=north,font=\scriptsize] {$1$};
			\draw ({3.2*(3/4)},0) -- ({3.2*(3/4)},{-3pt}) node[anchor=north,font=\scriptsize] {$1.075$};
			\draw ({3.2*(4/4)},0) -- ({3.2*(4/4)},{-3pt}) node[anchor=north,font=\scriptsize] {$1.15$};
			\node[rotate=90,font=\scriptsize] at (-.85,1.5) {Pricing kernel};
			\node[font=\scriptsize] at (1.6,-.8) {Market payoff};
			\draw[blue,thick] plot file {pricingkernel.txt};
		\end{tikzpicture}
		\begin{tikzpicture}
			\draw[thick, ->] (0,0) -- (3.5,0);
			\draw[thick, ->] (0,0) -- (0,3.3);
			\draw (0,0) -- ({-3pt},0) node[anchor=east,font=\scriptsize] {$.85$};
			\draw (0,{3*(1/4)}) -- ({-3pt},{3*(1/4)}) node[anchor=east,font=\scriptsize] {$.925$};
			\draw (0,{3*(2/4)}) -- ({-3pt},{3*(2/4)}) node[anchor=east,font=\scriptsize] {$1$};
			\draw (0,{3*(3/4)}) -- ({-3pt},{3*(3/4)}) node[anchor=east,font=\scriptsize] {$1.075$};
			\draw (0,{3*(4/4)}) -- ({-3pt},{3*(4/4)}) node[anchor=east,font=\scriptsize] {$1.15$};
			\draw (0,0) -- (0,{-3pt}) node[anchor=north,font=\scriptsize] {$.85$};
			\draw ({3.2*(1/4)},0) -- ({3.2*(1/4)},{-3pt}) node[anchor=north,font=\scriptsize] {$.925$};
			\draw ({3.2*(2/4)},0) -- ({3.2*(2/4)},{-3pt}) node[anchor=north,font=\scriptsize] {$1$};
			\draw ({3.2*(3/4)},0) -- ({3.2*(3/4)},{-3pt}) node[anchor=north,font=\scriptsize] {$1.075$};
			\draw ({3.2*(4/4)},0) -- ({3.2*(4/4)},{-3pt}) node[anchor=north,font=\scriptsize] {$1.15$};
			\node[rotate=90,font=\scriptsize] at (-1.15,1.5) {Derivative payoff};
			\node[font=\scriptsize] at (1.6,-.8) {Market payoff};
			\draw[red,dashed] (0,0) -- ({((1.11-.85)/(1.15-.85))*3.2},{((1.11-.85)/(1.15-.85))*3});
			\draw[blue,thick] plot file {ompd.txt};
		\end{tikzpicture}	
	\caption{Estimated probability densities for $\mu$ (solid) and $\nu$ (dashed) reported in \citet{Jackwerth00}, and the associated pricing kernel and optimal measure preserving derivative.}\label{fig:jackwerth}
\end{figure}
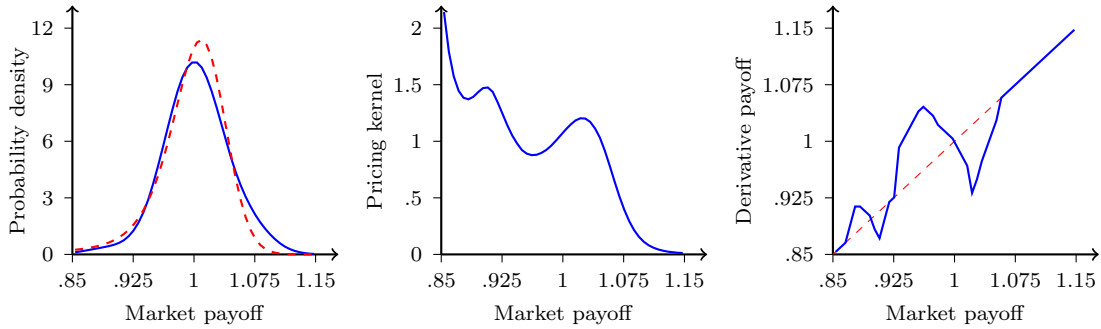

The probability measure $\mu$ represented by the density function plotted in the left panel in Figure \ref{fig:jackwerth} is atomless by construction, and therefore adequate. We therefore know from Proposition \ref{thm:adequate} that the optimal measure preserving derivative $\vartheta$ plotted in the right panel in Figure \ref{fig:jackwerth} minimizes the price $p(\theta)$ over $\theta\in\Theta$ subject to any of the constraints $\theta(X)\sim X$, $\theta(X)\gtrsim_1X$, $\theta(X)\gtrsim_\mathrm{cv}X$ and $\theta(X)\gtrsim_2X$. On the other hand, if $\mu$ was not adequate then these four constraints could potentially give rise to different price minimizing derivatives. To investigate this possibility further we replace $\mu$ with a discrete approximation to $\mu$ not satisfying the adequacy condition. Specifically, we partition the interval $(.85,1.15)$ into $n$ subintervals of equal length, and construct a discrete probability measure $\mu_n$ with an atom at the centre of each subinterval, that atom having mass equal to the probability assigned by $\mu$ to the subinterval. (The very small amount of probability assigned by $\mu$ to values less than .85 or greater than 1.15 is assigned to the atoms in the extreme left and right subintervals.) The top three panels in Figure \ref{fig:illustration} display the probability mass function for $\mu_n$ with $n=5,10,20$ (left, center, right). The mass functions are scaled to be directly comparable to the probability density function for $\mu$, displayed as a dotted line.

In order to suitably define the prices of derivatives when the market payoff is distributed according to $\mu_n$ we require a discrete approximation $\nu_n$ to $\nu$ which concentrates on the atoms of $\mu_n$. We choose this approximation to preserve the pricing kernel $\pi$ plotted in the center panel in Figure \ref{fig:jackwerth}. That is, we set $\nu_n\{x_i\}=\pi(x_i)\mu_n\{x_i\}$ for each $i=1,\dots,n$, where $x_1,\dots,x_n$ are the points to which $\mu_n$ assigns positive mass. The corresponding price of a derivative $\theta\in\Theta$ is then given by $p_n(\theta)=\sum_{i=1}^n\nu_n\{x_i\}\theta(x_i)$.

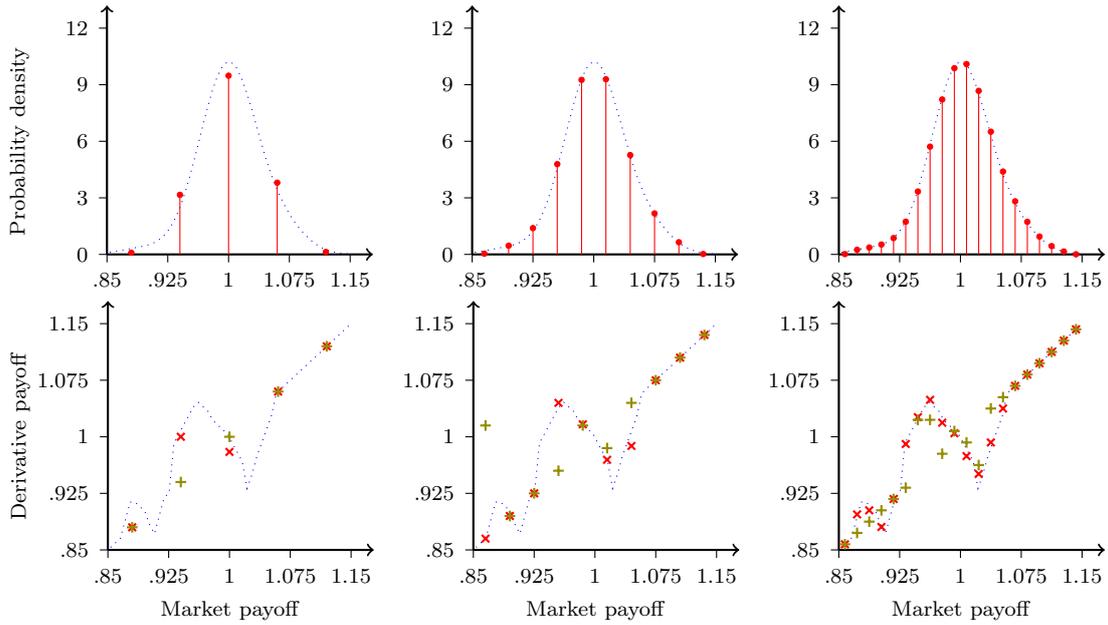
\begin{figure}
	\flushright
		\begin{tikzpicture}
			\draw[thick, ->] (0,0) -- (3.5,0);
			\draw[thick, ->] (0,0) -- (0,3.3);
			\draw (0,0) -- ({-3pt},0) node[anchor=east,font=\scriptsize] {$0$};
			\draw (0,{3*(1/4)}) -- ({-3pt},{3*(1/4)}) node[anchor=east,font=\scriptsize] {$3$};
			\draw (0,{3*(2/4)}) -- ({-3pt},{3*(2/4)}) node[anchor=east,font=\scriptsize] {$6$};
			\draw (0,{3*(3/4)}) -- ({-3pt},{3*(3/4)}) node[anchor=east,font=\scriptsize] {$9$};
			\draw (0,{3*(4/4)}) -- ({-3pt},{3*(4/4)}) node[anchor=east,font=\scriptsize] {$12$};
			\draw (0,0) -- (0,{-3pt}) node[anchor=north,font=\scriptsize] {$.85$};
			\draw ({3.2*(1/4)},0) -- ({3.2*(1/4)},{-3pt}) node[anchor=north,font=\scriptsize] {$.925$};
			\draw ({3.2*(2/4)},0) -- ({3.2*(2/4)},{-3pt}) node[anchor=north,font=\scriptsize] {$1$};
			\draw ({3.2*(3/4)},0) -- ({3.2*(3/4)},{-3pt}) node[anchor=north,font=\scriptsize] {$1.075$};
			\draw ({3.2*(4/4)},0) -- ({3.2*(4/4)},{-3pt}) node[anchor=north,font=\scriptsize] {$1.15$};
			\node[rotate=90,font=\scriptsize] at (-1.15,1.5) {Probability density};
			\draw[blue,dotted] plot file {physical.txt};
			\draw[red,only marks] plot[ycomb,mark=*,mark size=1pt] file {phys5.txt};
		\end{tikzpicture}
	    \hspace{.24cm}
		\begin{tikzpicture}
			\draw[thick, ->] (0,0) -- (3.5,0);
			\draw[thick, ->] (0,0) -- (0,3.3);
			\draw (0,0) -- ({-3pt},0) node[anchor=east,font=\scriptsize] {$0$};
			\draw (0,{3*(1/4)}) -- ({-3pt},{3*(1/4)}) node[anchor=east,font=\scriptsize] {$3$};
			\draw (0,{3*(2/4)}) -- ({-3pt},{3*(2/4)}) node[anchor=east,font=\scriptsize] {$6$};
			\draw (0,{3*(3/4)}) -- ({-3pt},{3*(3/4)}) node[anchor=east,font=\scriptsize] {$9$};
			\draw (0,{3*(4/4)}) -- ({-3pt},{3*(4/4)}) node[anchor=east,font=\scriptsize] {$12$};
			\draw (0,0) -- (0,{-3pt}) node[anchor=north,font=\scriptsize] {$.85$};
			\draw ({3.2*(1/4)},0) -- ({3.2*(1/4)},{-3pt}) node[anchor=north,font=\scriptsize] {$.925$};
			\draw ({3.2*(2/4)},0) -- ({3.2*(2/4)},{-3pt}) node[anchor=north,font=\scriptsize] {$1$};
			\draw ({3.2*(3/4)},0) -- ({3.2*(3/4)},{-3pt}) node[anchor=north,font=\scriptsize] {$1.075$};
			\draw ({3.2*(4/4)},0) -- ({3.2*(4/4)},{-3pt}) node[anchor=north,font=\scriptsize] {$1.15$};
			\draw[blue,dotted] plot file {physical.txt};
			\draw[red,only marks] plot[ycomb,mark=*,mark size=1pt] file {phys10.txt};
		\end{tikzpicture}
	    \hspace{.26cm}
		\begin{tikzpicture}
			\draw[thick, ->] (0,0) -- (3.5,0);
			\draw[thick, ->] (0,0) -- (0,3.3);
			\draw (0,0) -- ({-3pt},0) node[anchor=east,font=\scriptsize] {$0$};
			\draw (0,{3*(1/4)}) -- ({-3pt},{3*(1/4)}) node[anchor=east,font=\scriptsize] {$3$};
			\draw (0,{3*(2/4)}) -- ({-3pt},{3*(2/4)}) node[anchor=east,font=\scriptsize] {$6$};
			\draw (0,{3*(3/4)}) -- ({-3pt},{3*(3/4)}) node[anchor=east,font=\scriptsize] {$9$};
			\draw (0,{3*(4/4)}) -- ({-3pt},{3*(4/4)}) node[anchor=east,font=\scriptsize] {$12$};
			\draw (0,0) -- (0,{-3pt}) node[anchor=north,font=\scriptsize] {$.85$};
			\draw ({3.2*(1/4)},0) -- ({3.2*(1/4)},{-3pt}) node[anchor=north,font=\scriptsize] {$.925$};
			\draw ({3.2*(2/4)},0) -- ({3.2*(2/4)},{-3pt}) node[anchor=north,font=\scriptsize] {$1$};
			\draw ({3.2*(3/4)},0) -- ({3.2*(3/4)},{-3pt}) node[anchor=north,font=\scriptsize] {$1.075$};
			\draw ({3.2*(4/4)},0) -- ({3.2*(4/4)},{-3pt}) node[anchor=north,font=\scriptsize] {$1.15$};
			\draw[blue,dotted] plot file {physical.txt};
			\draw[red,only marks] plot[ycomb,mark=*,mark size=1pt] file {phys20.txt};
		\end{tikzpicture}	
		\begin{tikzpicture}
			\draw[thick, ->] (0,0) -- (3.5,0);
			\draw[thick, ->] (0,0) -- (0,3.3);
			\draw (0,0) -- ({-3pt},0) node[anchor=east,font=\scriptsize] {$.85$};
			\draw (0,{3*(1/4)}) -- ({-3pt},{3*(1/4)}) node[anchor=east,font=\scriptsize] {$.925$};
			\draw (0,{3*(2/4)}) -- ({-3pt},{3*(2/4)}) node[anchor=east,font=\scriptsize] {$1$};
			\draw (0,{3*(3/4)}) -- ({-3pt},{3*(3/4)}) node[anchor=east,font=\scriptsize] {$1.075$};
			\draw (0,{3*(4/4)}) -- ({-3pt},{3*(4/4)}) node[anchor=east,font=\scriptsize] {$1.15$};
			\draw (0,0) -- (0,{-3pt}) node[anchor=north,font=\scriptsize] {$.85$};
			\draw ({3.2*(1/4)},0) -- ({3.2*(1/4)},{-3pt}) node[anchor=north,font=\scriptsize] {$.925$};
			\draw ({3.2*(2/4)},0) -- ({3.2*(2/4)},{-3pt}) node[anchor=north,font=\scriptsize] {$1$};
			\draw ({3.2*(3/4)},0) -- ({3.2*(3/4)},{-3pt}) node[anchor=north,font=\scriptsize] {$1.075$};
			\draw ({3.2*(4/4)},0) -- ({3.2*(4/4)},{-3pt}) node[anchor=north,font=\scriptsize] {$1.15$};
			\node[rotate=90,font=\scriptsize] at (-1.15,1.5) {Derivative payoff};
			\node[font=\scriptsize] at (1.6,-.8) {Market payoff};
			\draw[blue,dotted] plot file {ompd.txt};
			\draw[red,only marks,thick] plot[mark=x,mark size=2pt] file {ompd5.txt};
			\draw[olive,only marks,thick] plot[mark=+,mark size=2pt] file {ompd5int.txt};
		\end{tikzpicture}
		\begin{tikzpicture}
			\draw[thick, ->] (0,0) -- (3.5,0);
			\draw[thick, ->] (0,0) -- (0,3.3);
			\draw (0,0) -- ({-3pt},0) node[anchor=east,font=\scriptsize] {$.85$};
			\draw (0,{3*(1/4)}) -- ({-3pt},{3*(1/4)}) node[anchor=east,font=\scriptsize] {$.925$};
			\draw (0,{3*(2/4)}) -- ({-3pt},{3*(2/4)}) node[anchor=east,font=\scriptsize] {$1$};
			\draw (0,{3*(3/4)}) -- ({-3pt},{3*(3/4)}) node[anchor=east,font=\scriptsize] {$1.075$};
			\draw (0,{3*(4/4)}) -- ({-3pt},{3*(4/4)}) node[anchor=east,font=\scriptsize] {$1.15$};
			\draw (0,0) -- (0,{-3pt}) node[anchor=north,font=\scriptsize] {$.85$};
			\draw ({3.2*(1/4)},0) -- ({3.2*(1/4)},{-3pt}) node[anchor=north,font=\scriptsize] {$.925$};
			\draw ({3.2*(2/4)},0) -- ({3.2*(2/4)},{-3pt}) node[anchor=north,font=\scriptsize] {$1$};
			\draw ({3.2*(3/4)},0) -- ({3.2*(3/4)},{-3pt}) node[anchor=north,font=\scriptsize] {$1.075$};
			\draw ({3.2*(4/4)},0) -- ({3.2*(4/4)},{-3pt}) node[anchor=north,font=\scriptsize] {$1.15$};
			\node[font=\scriptsize] at (1.6,-.8) {Market payoff};
			\draw[blue,dotted] plot file {ompd.txt};
			\draw[red,only marks,thick] plot[mark=x,mark size=2pt] file {ompd10.txt};
			\draw[olive,only marks,thick] plot[mark=+,mark size=2pt] file {ompd10int.txt};
		\end{tikzpicture}
		\begin{tikzpicture}
			\draw[thick, ->] (0,0) -- (3.5,0);
			\draw[thick, ->] (0,0) -- (0,3.3);
			\draw (0,0) -- ({-3pt},0) node[anchor=east,font=\scriptsize] {$.85$};
			\draw (0,{3*(1/4)}) -- ({-3pt},{3*(1/4)}) node[anchor=east,font=\scriptsize] {$.925$};
			\draw (0,{3*(2/4)}) -- ({-3pt},{3*(2/4)}) node[anchor=east,font=\scriptsize] {$1$};
			\draw (0,{3*(3/4)}) -- ({-3pt},{3*(3/4)}) node[anchor=east,font=\scriptsize] {$1.075$};
			\draw (0,{3*(4/4)}) -- ({-3pt},{3*(4/4)}) node[anchor=east,font=\scriptsize] {$1.15$};
			\draw (0,0) -- (0,{-3pt}) node[anchor=north,font=\scriptsize] {$.85$};
			\draw ({3.2*(1/4)},0) -- ({3.2*(1/4)},{-3pt}) node[anchor=north,font=\scriptsize] {$.925$};
			\draw ({3.2*(2/4)},0) -- ({3.2*(2/4)},{-3pt}) node[anchor=north,font=\scriptsize] {$1$};
			\draw ({3.2*(3/4)},0) -- ({3.2*(3/4)},{-3pt}) node[anchor=north,font=\scriptsize] {$1.075$};
			\draw ({3.2*(4/4)},0) -- ({3.2*(4/4)},{-3pt}) node[anchor=north,font=\scriptsize] {$1.15$};
			\node[font=\scriptsize] at (1.6,-.8) {Market payoff};
			\draw[blue,dotted] plot file {ompd.txt};
			\draw[red,only marks,thick] plot[mark=x,mark size=2pt] file {ompd20.txt};
			\draw[olive,only marks,thick] plot[mark=+,mark size=2pt] file {ompd20int.txt};
		\end{tikzpicture}	
	\caption{Price minimizing derivatives constrained to satisfy first-order ($+$) or second-order ($\times$) stochastic dominance over the market portfolio, based on the discrete approximation $\mu_n$ to $\mu$ with $n=5,10,20$ (left, center, right).}\label{fig:illustration}
\end{figure}

With the discrete measures $\mu_n$ and $\nu_n$ in hand, we compute the corresponding price minimizing derivatives constrained to satisfy first- or second-order stochastic dominance over the market portfolio. Computation was accomplished by solving a linear program in the case of second-order stochastic dominance, or a mixed-integer linear program in the case of first-order stochastic dominance. Such programs for minimizing linear objective functions subject to stochastic dominance constraints were developed in \citet{DentchevaRuszczynski03}, \citet{Kuosmanen04} and \citet{Luedtke08}; see also closely related results in \citet{RockafellarUryasev00} and \citet{Post03}, and recent applications in \citet{PostLongarela21} and \citet{BeareSeo22}. The bottom three panels in Figure \ref{fig:illustration} display the price minimizing derivatives we computed for $\mu_n$ and $\nu_n$ with $n=5,10,20$ (left, center, right). The payoffs for the derivatives constrained to satisfy first-order stochastic dominance are displayed with $+$ symbols, while those for the derivatives constrained to satisfy second-order stochastic dominance are displayed with $\times$ symbols. When the two symbols appear at the same location they resemble an asterisk. The three panels also display the optimal measure preserving derivative $\vartheta$ corresponding to the measures $\mu$ and $\nu$ as a dotted line.

The primary takeaway from Figure \ref{fig:illustration} is that, as $n$ increases from $5$ to $10$ to $20$, the price minimizing derivative computed for $\mu_n$ and $\nu_n$ under either the first- or second-order stochastic dominance constraint comes to closely resemble the optimal measure preserving derivative $\vartheta$ for $\mu$ and $\nu$, particularly in the case of the second-order stochastic dominance constraint. This indicates that Proposition \ref{thm:adequate}, which applies under the condition of adequacy, can also deliver a reasonable approximation to the price minimizing derivative constrained to satisfy first- or second-order stochastic dominance in contexts where the adequacy condition is only approximately satisfied.

The bottom-left panel in Figure \ref{fig:illustration} provides a nice illustration of Proposition \ref{thm:pkSSD}, which may be validly applied irrespective of whether the adequacy condition is met. Proposition \ref{thm:pkSSD} asserts that the price minimizing derivative constrained to satisfy second-order stochastic dominance deviates from the market portfolio if and only if the pricing kernel is nonmonotone. The pricing kernel corresponding to the bottom-left panel in Figure \ref{fig:illustration} is higher at $x_3$ than at $x_2$, and consequently nonmonotone. Consistent with Proposition \ref{thm:pkSSD}, we see that the payoffs for the price minimizing derivative constrained to satisfy second-order stochastic dominance deviate from those of the market portfolio, being higher at $x_2$ than at $x_3$. On the other hand, the payoffs for the price minimizing derivative constrained to satisfy first-order stochastic dominance lie precisely on the 45-degree line and therefore do not deviate from those of the market portfolio. This reflects the fact that, in the absence of the adequacy condition, pricing kernel nonmonotonicity is equivalent to the existence of stochastic arbitrage opportunities defined in the sense of second-order stochastic dominance, but is not equivalent to the existence of stochastic arbitrage opportunities defined in the sense of first-order stochastic dominance.

\section{Further discussion}\label{sec:discussion}

The connection between pricing kernel monotonicity and the existence of stochastic arbitrage opportunities has been an active area of research since the documentation of empirical violations of pricing kernel monotonicity with S\&P500 index options in \citet{AitSahaliaLo00}, \citet{Jackwerth00} and \citet{RosenbergEngle02}. It was immediately recognized in these articles that such violations are inconsistent with the existence of a risk-averse agent willing to hold the market portfolio. The empirical failure of pricing kernel monotonicity became known as the pricing kernel puzzle or risk aversion puzzle. In \citet{Beare11} I drew attention to the relevance of the results in \citet{Dybvig88} for understanding the implications of pricing kernel nonmonotonicity. The puzzling implication of these results was that, even in the absence of risk aversion, an agent merely preferring more to less will not hold the market portfolio under pricing kernel nonmonotonicity. This brought into question the centrality of risk aversion in understanding pricing kernel nonmonotonicity. See \citet[pp.\ 223--226]{Perrakis19} for further discussion of this history.

The results in \citet{PostLongarela21} and \citet{KleinerMoldovanuStrack21} discussed in this article establish that risk aversion is indeed the central issue relevant to understanding pricing kernel nonmonotonicity. Propositions \ref{thm:pkSSD} and \ref{thm:adequate} and Examples \ref{ex1} and \ref{ex2} in this article establish that it is the technical condition of adequacy imposed in \citet{Dybvig88} which broadens the set of stochastic arbitrage opportunities available under pricing kernel nonmonotonicity to include replicating portfolios.

The model we have studied is one of a complete and frictionless options market, and therefore highly stylized. Real-world options markets are neither complete nor frictionless. Empirical studies seeking to identify stochastic arbitrage opportunities in options markets, including studies of the market for S\&P500 index put and call options reported in \citet{ConstantinidesJackwerthPerrakis09}, \citet{PostLongarela21} and \citet{BeareSeo22}, take into account the fact that options are written only at limited strikes and in limited quantities (so that the market is incomplete), and the fact that bid-ask spreads can be wide (so that the market is not frictionless). The latter article reports the outcome of a search for stochastic arbitrage opportunities defined in two senses: first- and second-order stochastic dominance over the market portfolio. It is notable that the results differ greatly depending on which sense of stochastic arbitrage is adopted, with stochastic arbitrage opportunities defined in terms of second-order stochastic dominance being more frequently detected. This finding contrasts with our demonstration in Proposition \ref{thm:adequate} of the equivalence of the existence of either form of stochastic arbitrage under adequacy. Market incompleteness and frictions may be important factors in explaining this discrepancy.

Out of the four stochastic orders we have used to define different senses of stochastic arbitrage, second-order stochastic dominance is particularly appealing for a number of reasons. As the weakest of the stochastic orders, second-order stochastic dominance provides the greatest scope for identifying stochastic arbitrage opportunities. Moreover, any such opportunity ought to be preferred to the market portfolio by any investor with preferences satisfying the minimal conditions of nonsatiation and risk aversion. As discussed in \citet{KopaPost09}, the nonexistence of stochastic arbitrage opportunities defined in the sense of second-order dominance provides a natural characterization of market efficiency, whereas a distinction between admissibility and optimality makes the characterization of market efficiency more ambiguous with first-order stochastic dominance. The equivalence of pricing kernel nonmonotonicity and the existence of stochastic arbitrage opportunities defined in terms of second-order stochastic dominance does not hinge on the technical condition of adequacy, unlike the stronger first-order stochastic dominance and distributional equality relations. In practical implementations, stochastic arbitrage opportunities defined in the sense of second-order stochastic dominance can be easily computed using linear programs, whereas much more computationally demanding mixed-integer linear programs are needed in the case of first-order stochastic dominance or distributional equality.

Broadening the definition of stochastic arbitrage even further by adopting a weaker ordering such as third-order stochastic dominance may be advantageous in some respects. A security third-order stochastically dominating the market portfolio will be preferred by all nonsatiated risk-averse investors with preferences satisfying decreasing absolute risk aversion, a condition widely regarded as uncontroversial by financial economists. \citet{PostKopa17} have developed computational methods for identifying stochastic arbitrage opportunities defined in the sense of third-order stochastic dominance, while methods applicable under higher-order stochastic dominance constraints have been developed in \citet{FangPost22}. Future research may investigate the extent to which the results in this article linking the shape of the pricing kernel to the existence of stochastic arbitrage opportunities can be extended to accommodate more general notions of stochastic arbitrage defined in terms of higher-order stochastic dominance constraints.

\appendix

\numberwithin{theorem}{section}
\numberwithin{lemma}{section}
\numberwithin{equation}{section}
\numberwithin{thm}{section}

\section{Mathematical appendix}\label{sec:appx}

\subsection{Mathematical preliminaries}\label{sec:mathprelim}

\subsubsection{Distributions, quantiles, and rearrangements}\label{appx:rearrangement}

Let $(\Omega,\mu)$ be a probability space, and let $M_+$ denote the collection of all nonnegative real valued Borel measurable maps on $(\Omega,\mu)$. If $f\in M_+$ then we define the \emph{distribution function} of $f$ to be the map $F_f:\bar{\mathbb R}_+\to[0,1]$ given by
\begin{align*}
	F_f(x)&=\mu\{\omega\in\Omega:f(\omega)\leq x\},
\end{align*}
and we define the \emph{quantile function} of $f$ to be the map $Q_f:[0,1]\to\bar{\mathbb R}_+$ given by
\begin{align*}
	Q_f(u)&=\inf\{x\in\mathbb R_+:F_f(x)\geq u\}.
\end{align*}
The quantile function of $f$, or its right-continuous version, is also referred to as the \emph{nondecreasing rearrangement} of $f$.
\begin{lemma}[Generalized inverses]\label{lem:inverse}
	If $f\in M_+$ then $Q_f(F_f(x))\leq x$ for all $x\in\bar{\mathbb R}_+$. If also $F_f$ is continuous then $F_f(Q_f(u))=u$ for all $u\in[0,1]$.
\end{lemma}
\begin{proof}
	See parts (ii) and (iv) of Lemma 21.1 in \citet{vanderVaart98}.
\end{proof}
\begin{lemma}[Probability integral transforms]\label{lem:PIT}
	If $f\in M_+$ and if $u$ belongs to the range of $F_f$, then $\mu\{\omega\in\Omega:F_f(f(\omega))\leq u\}=u$.
\end{lemma}
\begin{proof}
	See Lemma 1 in \citet{Angus94}.
\end{proof}
\begin{lemma}[Second-order stochastic dominance]\label{lem:SSD}
	For $f_1,f_2\in M_+$ the following three conditions are equivalent.
	\begin{enumerate}
		\item $\int_0^yF_{f_1}(x)\mathrm{d}x\leq\int_0^yF_{f_2}(x)\mathrm{d}x$ for all $y\in\mathbb R_+$.\label{SSDf}
		\item $\int_0^vQ_{f_1}(u)\mathrm{d}u\geq\int_0^vQ_{f_2}(u)\mathrm{d}u$ for all $v\in[0,1]$.\label{SSDq}
		\item $\int_0^\infty u(f_1(x))\mathrm{d}\mu(x)\geq\int_0^\infty u(f_2(x))\mathrm{d}\mu(x)$ for all nondecreasing and concave maps $u:\mathbb R_+\to\mathbb R$ such that the integrals exist.\label{SSDu}
	\end{enumerate}
\end{lemma}
\begin{proof}
	See Theorems 4.A.1, 4.A.2 and 4.A.3 in \citet{ShakedShanthikumar07}.
\end{proof}
\begin{lemma}[Hardy quantile majorization inequality]\label{lem:Hardy}
	If $f_1,f_2\in M_+$ have quantile functions $Q_{f_1},Q_{f_2}$ satisfying condition \ref{SSDq} in Lemma \ref{lem:SSD}, and if $g:[0,1]\to\bar{\mathbb R}_+$ is nonincreasing, then
	\begin{align*}
		\int_0^vQ_{f_1}(u)g(u)\mathrm{d}u&\geq\int_0^vQ_{f_2}(u)g(u)\mathrm{d}u\quad\text{for all }v\in[0,1].
	\end{align*}
\end{lemma}
\begin{proof}
See Theorem 5.1 in \citet{Luxemburg67}, Theorem 7.1 in \citet{Day70} or Theorem 9.1 in \citet{ChongRice71} for the case where $Q_{f_1}$, $Q_{f_2}$, $Q_{f_1}g$ and $Q_{f_2}g$ are integrable. Since $Q_{f_1}$, $Q_{f_2}$ and $g$ are nonnegative we may dispense with integrability by applying the monotone convergence theorem.
\end{proof}
We say that maps $f_1,f_2\in M_+$ are \emph{comonotonic} if
\begin{align*}
	\mu\otimes\mu\{(\omega,\omega')\in\Omega\times\Omega:(f_1(\omega)-f_1(\omega'))(f_2(\omega)-f_2(\omega'))\geq0\}&=1,
\end{align*}
or \emph{countermonotonic} if
\begin{align*}
	\mu\otimes\mu\{(\omega,\omega')\in\Omega\times\Omega:(f_1(\omega)-f_1(\omega'))(f_2(\omega)-f_2(\omega'))\leq0\}&=1.
\end{align*}
\begin{lemma}[Hardy-Littlewood rearrangement inequality]\label{lem:HardyLittlewood}
	If $f_1,f_2\in M_+$ then
	\begin{align*}
		\int_0^1Q_{f_1}(u)Q_{f_2}(1-u)\mathrm{d}u&\leq\int f_1f_2\mathrm{d}\mu\leq\int_0^1Q_{f_1}(u)Q_{f_2}(u)\mathrm{d}u.
	\end{align*}
The former inequality holds with equality if $f_1$ and $f_2$ are countermonotonic, while the latter inequality holds with equality if $f_1$ and $f_2$ are comonotonic.
\end{lemma}
\begin{proof}
	For the inequalities, see Theorem 8.2 in \citet{Luxemburg67}, Theorem 10.1 in \citet{Day70} or Theorem 12.2 in \citet{ChongRice71},  and the subsequent remarks therein regarding integrability and nonnegativity. If $f_1$ and $f_2$ are countermonotonic then Theorem 2.5.5 in \citet{Nelsen06} implies that the joint distribution of $f_1$ and $f_2$ is its Fr\'{e}chet-Hoeffding lower bound, and thus Theorem 3.1 in \citet{PuccettiWang15} implies that the joint distribution of $f_1$ and $f_2$ is the same as that of $Q_{f_1}(U)$ and $Q_{f_2}(1-U)$, where $U$ is a random variable distributed uniformly on $(0,1)$. Thus the former inequality holds with equality. If $f_1$ and $f_2$ are comonotonic then a symmetric argument using Theorem 2.5.4 in \citet{Nelsen06} and Theorem 2.1 in \citet{PuccettiWang15} establishes that the latter inequality holds with equality.
\end{proof}

\subsubsection{Disintegration of probability measures}\label{appx:disintegration}

Let $(\Omega,\mu)$ and $M_+$ be as in Section \ref{appx:rearrangement}, and suppose further that $\Omega$ is a metric space and $\mu$ is Borel. Given any map $f\in M_+$, we will write $\mu f^{-1}$ for the Borel probability measure on $\mathbb R_+$ assigning measure $\mu\{x:f(x)\in B\}$ to each Borel subset $B$ of $\mathbb R_+$.
\begin{lemma}[Disintegration theorem]\label{thm:disintegration}
	Given any map $\pi\in M_+$, there exists a collection $\{\mu_z:z\in\mathbb R_+\}$ of Borel probability measures on $\Omega$ satisfying the following conditions.
	\begin{enumerate}
		\item $\mu_z\{x:\pi(x)=z\}=1$ for $\mu\pi^{-1}$-almost all $z$.\label{eq:disintegration1}
		\item $z\mapsto\int f\mathrm{d}\mu_z$ is measurable for each $f\in M_+$.\label{eq:disintegration2}
		\item $\int f\mathrm{d}\mu=\iint f\mathrm{d}\mu_{\pi(x)}\mathrm{d}\mu(x)$ for each $f\in M_+$.\label{eq:disintegration3}
	\end{enumerate}
\end{lemma}
\begin{proof}
	See Theorems 1 and 2 in \citet{ChangPollard97}.
\end{proof}

\subsection{Proofs of Propositions \ref{thm:pkSSD} and \ref{thm:adequate}}\label{sec:proofs}

\begin{lemma}\label{lem:LowerBound}
	If $\mu$ and $\nu$ are mutually absolutely continuous with Radon-Nikodym derivative $\pi=\mathrm{d}\nu/\mathrm{d}\mathrm{\mu}$, then every $\theta\in\Theta$ satisfying $\theta(X)\gtrsim_2X$ also satisfies
	\begin{align*}
		\int\theta\mathrm{d}\nu&\geq\int_0^1Q(u)Q_{\pi}(1-u)\mathrm{d}u.
	\end{align*}
\end{lemma}

\begin{proof}
	The Hardy-Littlewood rearrangement inequality (Lemma \ref{lem:HardyLittlewood}) implies that, for any $\theta\in\Theta$,
	\begin{align}\label{eq:price1}
		\int\theta\mathrm{d}\nu&=\int\theta\pi\mathrm{d}\mu\geq\int_0^1Q_{\theta}(u)Q_{\pi}(1-u)\mathrm{d}u.
	\end{align}
	For any $\theta\in\Theta$ satisfying $\theta(X)\gtrsim_2X$ we have
	\begin{align*}
		\int_0^vQ_{\theta}(u)\mathrm{d}u\geq\int_0^vQ(u)\mathrm{d}u
	\end{align*}
	for all $v\in[0,1]$, by Lemma \ref{lem:SSD}. It therefore follows from the Hardy quantile majorization inequality (Lemma \ref{lem:Hardy}) that
	\begin{align}\label{eq:price2}
		&\int_0^1Q_{\theta}(u)Q_{\pi}(1-u)\mathrm{d}u\geq\int_0^1Q(u)Q_{\pi}(1-u)\mathrm{d}u
	\end{align}
	for all $\theta\in\Theta$ satisfying $\theta(X)\gtrsim_2X$. Combining \eqref{eq:price1} and \eqref{eq:price2} establishes the claimed inequality.
\end{proof}

\begin{proof}[Proof of Proposition \ref{thm:pkSSD}]
	We will show that \ref{pkSSD1} implies \ref{pkSSD2}, that \ref{pkSSD2} implies \ref{pkSSD3}, and that \ref{pkSSD3} implies \ref{pkSSD1}. The fact that \ref{pkSSD1} implies \ref{pkSSD2} is immediate from the fact that the relation $\gtrsim_2$ is implied by the relation $\gtrsim_\mathrm{cv}$. To show that \ref{pkSSD2} implies \ref{pkSSD3}, suppose that \ref{pkSSD3} is false. Then $x\mapsto x$ and $x\mapsto\pi(x)$ are countermonotonic, so that the lower bound in the Hardy-Littlewood rearrangement inequality (Lemma \ref{lem:HardyLittlewood}) holds with equality:
	\begin{align*}
		\int_0^1Q(u)Q_{\pi}(1-u)\mathrm{d}u&=\int x\pi(x)\mathrm{d}\mu(x)=\int x\mathrm{d}\nu(x).
	\end{align*}
	It therefore follows from Lemma \ref{lem:LowerBound} that \ref{pkSSD2} must be false. Thus \ref{pkSSD2} implies \ref{pkSSD3}.
	
	To show that \ref{pkSSD3} implies \ref{pkSSD1}, suppose that \ref{pkSSD3} is true. In this case we may choose Borel sets $A,B$ with $\mu(A),\mu(B)>0$ such that
	\begin{align*}
		\sup A<\inf B\,\,\,\text{and}\,\,\,\esssup_{x\in A}\pi(x)<\essinf_{x\in B}\pi(x).
	\end{align*}
	Given $\epsilon\in(0,\mu(B)\inf B)$, define a function $\theta_\epsilon\in\Theta$ by
	\begin{align*}
		\theta_\epsilon(x)&=\begin{cases}x+\frac{\epsilon}{\mu(A)}&\text{for }x\in A\\
			x-\frac{\epsilon}{\mu(B)}&\text{for }x\in B\\
			x&\text{otherwise.}\end{cases}
	\end{align*}
	Observe that
	\begin{align*}
		\int\theta_\epsilon\mathrm{d}\nu&=\int x\mathrm{d}\nu(x)+\frac{\epsilon}{\mu(A)}\nu(A)-\frac{\epsilon}{\mu(B)}\nu(B)\\
		&\leq\int x\mathrm{d}\nu(x)+\epsilon\esssup_{x\in A}\pi(x)-\epsilon\essinf_{x\in B}\pi(x)<\int x\mathrm{d}\nu(x).
	\end{align*}
	Thus \ref{pkSSD1} will be shown true if we can establish that $\theta_\epsilon(X)\gtrsim_\mathrm{cv}X$ for sufficiently small $\epsilon>0$. We may accomplish this by showing that, for any $\epsilon$ small enough that
	\begin{align}\label{eq:smalleps}
		\sup A+\frac{\epsilon}{\mu(A)}\leq\inf B-\frac{\epsilon}{\mu(B)},
	\end{align}
	we have $\int u(\theta_\epsilon(x))\mathrm{d}\mu(x)\geq\int u(x)\mathrm{d}\mu(x)$ for all concave maps $u:\mathbb R_+\to\mathbb R$ such that the integrals exist. Any such $u$ has a nonincreasing right-derivative $u'$ on $(0,\infty)$ and satisfies
	\begin{align*}
		u\left(x+\frac{\epsilon}{\mu(A)}\right)&\geq u(x)+\frac{\epsilon}{\mu(A)}u'\left(x+\frac{\epsilon}{\mu(A)}\right)
	\end{align*}
	for all $x\in A$, and
	\begin{align*}
		u\left(x-\frac{\epsilon}{\mu(B)}\right)&\geq u(x)-\frac{\epsilon}{\mu(B)}u'\left(x-\frac{\epsilon}{\mu(B)}\right)
	\end{align*}
	for all $x\in B$. Therefore,
	\begin{align*}
		\int u(\theta_\epsilon(x))\mathrm{d}\mu(x)&\geq\int u(x)\mathrm{d}\mu(x)+\frac{\epsilon}{\mu(A)}\int_A u'\left(x+\frac{\epsilon}{\mu(A)}\right)\mathrm{d}\mu(x)\\
		&\quad-\frac{\epsilon}{\mu(B)}\int_B u'\left(x-\frac{\epsilon}{\mu(B)}\right)\mathrm{d}\mu(x)\\
		&\geq\int u(x)\mathrm{d}\mu(x)+\epsilon u'\left(\sup A+\frac{\epsilon}{\mu(A)}\right)-\epsilon u'\left(\inf B-\frac{\epsilon}{\mu(B)}\right),
	\end{align*}
	using the fact that $u'$ is nonincreasing to obtain the second inequality. Relying again on the fact that $u'$ is nonincreasing, we deduce that $\int u(\theta_\epsilon(x))\mathrm{d}\mu(x)\geq\int u(x)\mathrm{d}\mu(x)$ if $\epsilon$ is chosen small enough to satisfy \eqref{eq:smalleps}. Thus \ref{pkSSD3} implies \ref{pkSSD1}.
\end{proof}

\begin{proof}[Proof of Proposition \ref{thm:adequate}]
	Proposition \ref{thm:pkSSD} establishes the equivalence of \ref{adequate3}, \ref{adequate4} and \ref{adequate5}, so it suffices for us to show that $\vartheta$ attains the infimum in \ref{adequate4} while also satisfying $\vartheta(X)\sim X$. We first show that $\vartheta(X)\sim X$, considering separately the cases where $\mu$ is atomless and where $\mu$ is adequate but not atomless. Suppose first that $\mu$ is atomless. In this case $F(Q(u))=u$ for all $u$ by Lemma \ref{lem:inverse} and so, for any $x,y\in\mathbb R_+$, if $\vartheta(x)\leq y$ then
	\begin{align}\label{eq:FQ}
		1-F_{\pi}(\pi(x))+\mu\{w:\pi(w)=\pi(x),w\leq x\}&\leq F(y).
	\end{align}
	Conversely, if \eqref{eq:FQ} is true then, using the fact that $Q(F(y))\leq y$ (Lemma \ref{lem:inverse}), we must have $\vartheta(x)\leq y$. We deduce that, for any $y\in\mathbb R_+$,
	\begin{align*}
		A_y&=\Big\{x:\mu\{w:\pi(w)=\pi(x),w\leq x\}\leq F(y)+F_{\pi}(\pi(x))-1\Big\},
	\end{align*}
	where we define $A_y=\{x:\vartheta(x)\leq y\}$. Our goal is to show that $\mu A_y=F(y)$.
	
	Applying the disintegration theorem (Lemma \ref{thm:disintegration}), we disintegrate $\mu$ into a family of probability measures $\{\mu_z:z\in\mathbb R_+\}$ on $\mathbb R_+$ satisfying conditions \ref{eq:disintegration1}, \ref{eq:disintegration2} and \ref{eq:disintegration3} in Lemma \ref{thm:disintegration}, with $\Theta$ playing the role of $M_+$. Define
	\begin{align*}
		B&=\{z:\mu\{w:\pi(w)=z\}>0\},
	\end{align*}
	a countable set. For any $z$ such that $\mu_z$ has an atom $\{v\}$, condition \ref{eq:disintegration3} in Lemma \ref{thm:disintegration} implies that
	\begin{align*}
		\mu\{v\}&=\int\mu_{\pi(w)}\{v\}\mathrm{d}\mu(w)\geq\mu\{w:\pi(w)=z\}\mu_z\{v\}.
	\end{align*}
	The final quantity is positive if $z\in B$, contradicting our supposition that $\mu$ is atomless. We deduce that, in the case where $\mu$ is atomless, we also have $\mu_z$ atomless for each $z\in B$.
	
	Condition \ref{eq:disintegration1} in Lemma \ref{thm:disintegration} implies that
	\begin{align}\label{eq:muzAy}
		\mu_zA_y&=\mu_z\Big\{x:\mu\{w:\pi(w)=z,w\leq x\}\leq F(y)+F_{\pi}(z)-1\Big\}
	\end{align}
	for $\mu\pi^{-1}$-almost all $z$. Moreover,
	\begin{align*}
		\mu\{w:\pi(w)=z,w\leq x\}&=\int\mu_{\pi(v)}\{w:\pi(w)=z,w\leq x\}\mathrm{d}\mu(v)\\
		&=\mu\{v:\pi(v)=z\}\mu_{z}\{w:w\leq x\},
	\end{align*}
	where the first equality follows from condition \ref{eq:disintegration3} in Lemma \ref{thm:disintegration} and the second from condition \ref{eq:disintegration1}. We may therefore rewrite \eqref{eq:muzAy} as
	\begin{align}\label{eq:muzAy2}
		\mu_zA_y&=\mu_z\left\{x:\mu\{w:\pi(w)=z\}\mu_{z}\{w:w\leq x\}\leq F(y)+F_{\pi}(z)-1\right\}.
	\end{align}
	Condition \ref{eq:disintegration3} in Lemma \ref{thm:disintegration} implies that
	\begin{align}\label{eq:muzAy3}
		\mu A_y&=\int\mu_zA_y\mathrm{d}\mu\pi^{-1}(z).
	\end{align}
	We will use \eqref{eq:muzAy2} and \eqref{eq:muzAy3} to show that $\mu A_y=F(y)$. In pursuit of this goal, it will be useful to separate the domain of integration in \eqref{eq:muzAy3} into the regions where $z\in B$ and where $z\notin B$.

	We first examine the integral over the region where $z\in B$. It follows from \eqref{eq:muzAy2} that, for each $z\in B$,
	\begin{align*}
		\mu_zA_y&=\mu_z\left\{x:\mu_{z}\{w:w\leq x\}\leq \frac{F(y)+F_{\pi}(z)-1}{\mu\{w:\pi(w)=z\}}\right\}.
	\end{align*}
	Since $\mu_z$ is atomless for each $z\in B$, the distribution function $x\mapsto\mu_z\{w:w\leq x\}$ corresponding to the probability measure $\mu_z$ is continuous for each $z\in B$. Lemma \ref{lem:PIT} on the uniformity of probability integral transforms thus ensures that
	\begin{align}
		\mu_zA_y&=0\vee\left(1\wedge\frac{F(y)+F_{\pi}(z)-1}{\mu\{w:\pi(w)=z\}}\right)\label{eq:maxmin}
	\end{align}
	for each $z\in B$. Noting that $F_{\pi}$ has a jump of size $\mu\{w:\pi(w)=z\}$ at each $z\in B$, we see that the fraction in \eqref{eq:maxmin} is equal to or less than zero if and only if $F_{\pi}(z)\leq1-F(y)$, and is greater than one if and only if $F_{\pi}(z-)>1-F(y)$.
	Let $B_y$ denote the set of all $z\in B$ such that $F_{\pi}(z-)\leq1-F(y)<F_{\pi}(z)$,	either empty or a singleton. Integrating \eqref{eq:maxmin} over $z\in B$ with respect to $\mu\pi^{-1}$ yields
	\begin{align}
		\int_{z\in B}\mu_zA_y\mathrm{d}\mu\pi^{-1}(z)&=\mu\pi^{-1}\{z\in B:F_{\pi}(z-)>1-F(y)\}\notag\\
		&\quad+\sum_{z\in B_y}(F(y)+F_{\pi}(z)-1).\label{eq:intB}
	\end{align}
	
	We next examine the integral over the region where $z\notin B$. For $\mu\pi^{-1}$-almost all $z\notin B$, we deduce from \eqref{eq:muzAy2} that
	\begin{align*}
		\mu_zA_y&=\begin{cases}0&\text{if }F_{\pi}(z)<1-F(y)\\
			1&\text{if }F_{\pi}(z)\geq1-F(y).\end{cases}
	\end{align*}
	Integrating over $z\notin B$ with respect to $\mu\pi^{-1}$, we obtain
	\begin{align*}
		\int_{z\notin B}\mu_{z}A_y\mathrm{d}\mu\pi^{-1}(z)&=\mu\pi^{-1}\{z\notin B:F_{\pi}(z)\geq1-F(y)\}.
	\end{align*}
	The set $\{z\notin B:F_{\pi}(z)=1-F(y)\}$ is $\mu\pi^{-1}$-null due to the fact that $F_{\pi}$ has no jumps in this set and takes the constant value $1-F(y)$. Since $F_{\pi}(z-)=F_{\pi}(z)$ at every $z\notin B$, it follows that
	\begin{align}
		\int_{z\notin B}\mu_{z}A_y\mathrm{d}\mu\pi^{-1}(z)&=\mu\pi^{-1}\{z\notin B:F_{\pi}(z-)>1-F(y)\}.\label{eq:intnotB}
	\end{align}

	Combining \eqref{eq:muzAy3}, \eqref{eq:intB} and \eqref{eq:intnotB}, we obtain
	\begin{align*}
		\mu A_y&=\mu\pi^{-1}\{z:F_{\pi}(z-)>1-F(y)\}+\sum_{z\in B_y}(F(y)+F_{\pi}(z)-1).
	\end{align*}
	If $B_y$ is empty then $1-F(y)$ belongs to the range of $F_{\pi}$, and so Lemma \ref{lem:PIT} implies that
	\begin{align*}
		\mu\pi^{-1}\{z:F_{\pi}(z)>1-F(y)\}&=F(y).
	\end{align*}
	Emptiness of $B_y$ also implies that $F_{\pi}(z)>1-F(y)$ if and only if $F_{\pi}(z-)>1-F(y)$, so we conclude that $\mu A_y=F(y)$ if $B_y$ is empty. On the other hand, if $B_y$ is a singleton, say $B_y=\{b\}$, then $F_{\pi}$ jumps from at most $1-F(y)$ to above $1-F(y)$ at $b$, so that
	\begin{align*}
		\mu\pi^{-1}\{z:F_{\pi}(z-)>1-F(y)\}&=1-F_{\pi}(b).
	\end{align*}
	Thus we also have $\mu A_y=F(y)$ if $B_y$ is a singleton. As noted above, $B_y$ must be either empty or a singleton, so we have $\mu A_y=F(y)$ for all $y$. This establishes that $\vartheta(X)\sim X$ in the case where $\mu$ is atomless.
	
	We next show that $\vartheta(X)\sim X$ in the case where $\mu$ is adequate but not atomless, meaning that $\mu$ allocates mass $n^{-1}$ to each of $n$ points $x_1<\cdots<x_n$. Let $m_1,\dots,m_n$ be the unique permutation of $1,\dots,n$ such that (i) $\pi(x_{m_i})\geq\pi(x_{m_j})$ whenever $i\leq j$, and (ii) $m_i<m_j$ whenever $\pi(x_{m_i})=\pi(x_{m_j})$ and $i>j$. For any $i$, we have $Q(i/n)=x_i$,
	\begin{align*}
		1-F_{\pi}(\pi(x_i))=\frac{1}{n}\sum_{j=1}^n\mathbbm{1}(\pi(x_j)>\pi(x_i))&=\frac{m_i}{n}-\frac{1}{n}\sum_{j=1}^{i}\mathbbm{1}(\pi(x_j)=\pi(x_i)),\quad\text{and}\\
		\mu\{w:\pi(w)=\pi(x_i),w\leq x_i\}&=\frac{1}{n}\sum_{j=1}^{i}\mathbbm{1}(\pi(x_j)=\pi(x_i)).
	\end{align*}
	Therefore, $\vartheta(x_i)=x_{m_i}$, meaning that $\vartheta$ is a permutation of the mass points of $\mu$. This establishes that $\vartheta(X)\sim X$ in the case where $\mu$ is adequate but not atomless.
	
	It remains to show that $\vartheta$ minimizes $\int\theta\mathrm{d}\nu$ over the set of all $\theta\in\Theta$ satisfying $\theta(X)\gtrsim_2X$. Lemma \ref{lem:LowerBound} establishes that all such $\theta$ satisfy
	\begin{align}\label{eq:priceLB}
		\int\theta\mathrm{d}\nu&\geq\int_0^1Q(u)Q_{\pi}(1-u)\mathrm{d}u.
	\end{align}
	We have shown that $\vartheta(X)\sim X$, which implies that $\vartheta(X)\gtrsim_2X$. We are therefore done if we can show that $\vartheta$ attains the lower bound in \eqref{eq:priceLB}.
	
	The Hardy-Littlewood rearrangement inequality (Lemma \ref{lem:HardyLittlewood}) implies that
	\begin{align}\label{eq:priceHL}
		\int\vartheta\mathrm{d}\nu&=\int\vartheta\pi\mathrm{d}\mu\geq\int_0^1Q_\vartheta(u)Q_\pi(1-u)\mathrm{d}u.
	\end{align}
	The inequality holds with equality if $\vartheta$ and $\pi$ are countermonotonic, which we now demonstrate. Pick $x,x'\in\mathbb R_+$ such that $\pi(x)>\pi(x')$. Since $F_{\pi}$ is nondecreasing and has a jump of size $\mu\{w:\pi(w)=\pi(x)\}$ at $\pi(x)$, we have
	\begin{align*}
		F_{\pi}(\pi(x))&\geq F_{\pi}(\pi(x'))+\mu\{w:\pi(w)=\pi(x)\},
	\end{align*}
	and consequently
	\begin{align*}
		F_{\pi}(\pi(x))-\mu\{w:\pi(w)=\pi(x),w\leq x\}&\geq F_{\pi}(\pi(x'))\\&\quad-\mu\{w:\pi(w)=\pi(x'),w\leq x'\}.
	\end{align*}
	It thus follows from the nondecreasing property of $Q$ that $\vartheta(x)\leq\vartheta(x')$. We conclude that $\vartheta$ and $\pi$ are countermonotonic. The inequality in \eqref{eq:priceHL} therefore holds with equality. Moreover, since $\vartheta(X)\sim X$, we have $Q_\vartheta=Q$. This shows that $\vartheta$ attains the lower bound in \eqref{eq:priceLB}.
\end{proof}

\end{document}